\documentclass[conference]{IEEEtran}
\IEEEoverridecommandlockouts

\usepackage{cite}
\usepackage{amsmath,amssymb,amsfonts,amsthm}
\newtheorem{theorem}{Theorem}
\newtheorem{proposition}[theorem]{Proposition}
\usepackage{algorithmic}
\usepackage{algorithm}
\usepackage{graphicx}
\usepackage{textcomp}
\usepackage{xcolor}
\usepackage{booktabs}
\usepackage{multirow}
\usepackage{url}
\usepackage{hyperref}
\usepackage{pifont}
\usepackage{dsfont}

\def\BibTeX{{\rm B\kern-.05em{\sc i\kern-.025em b}\kern-.08em
    T\kern-.1667em\lower.7ex\hbox{E}\kern-.125emX}}

\begin{document}

\title{OpenCLAW-Nexus: A Self-Reinforcing Trust Framework for Byzantine-Resilient Decentralized Federated Learning}

\author{
\IEEEauthorblockN{Wenyang Jia\textsuperscript{1}, Qiankang Xu\textsuperscript{1}, Ziwei Yan\textsuperscript{1}, Chunhua Kang\textsuperscript{3}, Yang Yang\textsuperscript{2}, Jinglu He\textsuperscript{2}, Kai Lei\textsuperscript{1,*}}
\IEEEauthorblockA{\textsuperscript{1}ICN Lab, Shenzhen Graduate School, Peking University\\}
\IEEEauthorblockA{\textsuperscript{2}Xunce Technology\\}
\IEEEauthorblockA{\textsuperscript{3}Shenzhen National High-Tech Industry Innovation Center\\}
\IEEEauthorblockA{\textsuperscript{*}Corresponding author: leik@pkusz.edu.cn}
}

\maketitle

\begin{abstract}
Decentralized Federated Learning (DFL) eliminates the central aggregator but introduces a severe ``trust gap'': without a trusted coordinator, the system becomes vulnerable to Byzantine and Sybil attacks, while existing solutions treat node selection, aggregation, and consensus as isolated modules, often relying on a trusted root dataset unavailable in truly decentralized settings.We propose \textsc{OpenCLAW-Nexus}, a self-reinforcing trust framework that bridges this gap through a single primitive, a discounted Beta-reputation model, that unifies reputation-based node selection, reputation-weighted aggregation (\emph{Rep-FedAvg}), and reputation-aware BFT consensus.
Rep-FedAvg eliminates the trusted root dataset requirement; we formally prove reputation separation between honest and Byzantine nodes under non-IID data with noisy evaluations.On a 1,000-node global testbed spanning three cloud providers and nine regions, Rep-FedAvg achieves 72.6\% accuracy on non-IID CIFAR-10 with 20\% Byzantine nodes and record-level differential privacy, within 0.5\,pp of centralized FLTrust.
Under a 300-node Sybil attack, reputation-weighted consensus maintains 84.2\% validation correctness versus 62.8\% (PoW) and 47.6\% (PoS).
\end{abstract}


\section{Introduction}

The proliferation of open-source decentralized AI platforms---exemplified by OpenCLAW~\cite{openclaw2025}, which has attracted a growing community of contributors across heterogeneous cloud environments---demonstrates the practical demand for serverless collaborative intelligence.
In such platforms, geographically distributed nodes voluntarily contribute computational resources to jointly train, evaluate, and deploy machine learning models without relying on a centralized coordinator.
Federated learning (FL)~\cite{mcmahan2017fedavg} provides the algorithmic foundation for this paradigm by enabling collaborative model training without exchanging raw data.

However, the dominant FL paradigm relies on a trusted central server to coordinate training rounds, aggregate model updates, and distribute the global model.
This architecture introduces a single point of failure and a trust bottleneck: participants must trust the server to aggregate honestly, while the server must trust participants to submit genuine updates.
In multi-organizational settings such as JointCloud computing and cross-institutional consortia, no single entity is universally trusted to serve as the central coordinator due to regulatory, competitive, and jurisdictional constraints.

The operational experience of deploying OpenCLAW across multiple cloud providers reveals that \emph{decentralized} FL~\cite{beutel2020flower, goethals2025delta}---while eliminating the central server---exposes three interdependent trust problems that the centralized paradigm implicitly delegates to the server:

\textbf{(P1) Participant discovery and selection.}
In the absence of a central registry, nodes must identify reliable training partners from an open, potentially adversarial network.
Random participant selection provides no defense against Byzantine and free-riding nodes~\cite{blanchard2017krum}.

\textbf{(P2) Aggregation trust.}
Without a trusted aggregator, the aggregation function must be resilient to malicious gradients, which can degrade or backdoor the global model~\cite{bagdasaryan2020}.
Existing Byzantine-robust aggregation methods either require a centralized root dataset (FLTrust~\cite{cao2021fltrust}) or discard potentially useful updates through hard filtering (Krum~\cite{blanchard2017krum}, Trimmed Mean~\cite{yin2018byzantine}).

\textbf{(P3) Model validation.}
Without a server-designated authoritative model version, nodes must reach distributed agreement on the current global model.
In a Byzantine setting, naive broadcasting permits attackers to propagate poisoned models to honest participants.

These three problems are \emph{interdependent}: improved participant selection yields cleaner aggregation inputs; validated model checkpoints feed back into participant evaluation.
However, existing work addresses them in isolation---Byzantine-robust aggregation ignores participant selection; decentralized FL frameworks assume honest peer discovery; consensus protocols assign equal weight to all voters regardless of track record.
This gap between the practical needs of systems like OpenCLAW and the state of the art motivates the present work.

\subsection{Proposed Approach}

We observe that all three problems reduce to a common \emph{trust estimation} problem.
Accordingly, we propose a single, continuously-updated \textbf{reputation score}---maintained per node and updated after every FL round---as a unified trust primitive that addresses all three challenges:

\begin{itemize}
    \item \textbf{Participant selection} (P1): Reputation-weighted scoring replaces random sampling, filtering out low-trust nodes prior to each FL round.
    \item \textbf{Aggregation weighting} (P2): Each node's model update is weighted by its reputation during aggregation (\emph{Rep-FedAvg}), enabling soft attenuation of Byzantine contributions without discarding honest updates.
    \item \textbf{Model validation} (P3): Reputation determines voting power in a graduated-quorum BFT consensus protocol, ensuring that newly created Sybil identities carry negligible influence on model acceptance decisions.
\end{itemize}

These stages form a \textbf{self-reinforcing cycle}: FL outcomes update reputation scores, which in turn drive participant selection and aggregation weights for subsequent rounds, thereby producing progressively higher-quality FL outcomes.

\subsection{Contributions}

The contributions of this work are as follows:

\begin{enumerate}
    \item \textbf{Rep-FedAvg}: A reputation-weighted federated aggregation algorithm with record-level differential privacy via local DP-SGD that achieves robustness comparable to centralized FLTrust without a server-owned private root dataset; instead, it uses a public validation benchmark and decentralized adjudication to update trust over time (Section~\ref{sec:repfedavg}).

    \item \textbf{Reputation-driven participant management}: A unified reputation mechanism that simultaneously governs participant selection, aggregation weighting, and model validation, forming a closed-loop trust cycle (Section~\ref{sec:pipeline}).

    \item \textbf{Graduated-quorum consensus for model validation}: A reputation-weighted BFT protocol with operation-specific thresholds that preserves safety guarantees under Sybil attacks (Section~\ref{sec:consensus}).

    \item \textbf{Large-scale empirical validation}: A comprehensive evaluation on \textbf{1{,}000 nodes across 3 cloud providers and 9~regions}, constituting the first multi-cloud deployment of reputation-driven decentralized FL at this scale (Section~\ref{sec:evaluation}).
\end{enumerate}

\smallskip
\noindent\textbf{Scope.}
This work targets supervised learning tasks with verifiable correctness metrics (e.g., image classification with a public validation benchmark and a separate held-out test set).
Extension to generative or open-ended tasks requires alternative adjudication mechanisms, which are discussed in Section~\ref{sec:limitations}.

\section{Related Work}\label{sec:related}

\textbf{Byzantine-robust FL.}
FedAvg~\cite{mcmahan2017fedavg} is the canonical aggregation baseline.
Krum~\cite{blanchard2017krum} selects the update closest to others; Trimmed Mean~\cite{yin2018byzantine} discards outlier coordinates; FLAME~\cite{nguyen2022flame} combines clustering with clipping.
FLTrust~\cite{cao2021fltrust} bootstraps trust via a server-side root dataset.
BALANCE~\cite{balance2024} enforces local similarity for decentralized Byzantine resilience.
DP-RSA~\cite{dprsa2022} jointly addresses differential privacy and robust aggregation.
All these address aggregation robustness but assume honest participant selection and centralized model distribution.

\textbf{Reputation in FL.}
FLARE~\cite{younesi2025flare} scores participants with multi-dimensional reputation and adaptive thresholds, but operates in a client--server architecture.
Fed-Credit~\cite{chen2024fedcredit} maintains credibility sets with time decay, though it similarly requires a centralized coordinator.
Murmura~\cite{rangwala2025murmura} applies evidential trust for decentralized FL on IoT wearables, but targets personalized models rather than a shared global model.
None integrates reputation with decentralized participant discovery and BFT model validation.

\textbf{Decentralized FL systems.}
Flower~\cite{beutel2020flower} provides FL primitives but relies on a central coordinator.
Delta Sum Learning~\cite{goethals2025delta} achieves fast gossip-based convergence without trust mechanisms.
Our work is distinguished by integrating trust into \emph{every stage} of the decentralized FL pipeline.

\textbf{P2P systems and BFT consensus.}
Kademlia~\cite{maymounkov2002kademlia} provides $O(\log n)$ DHT routing.
PBFT~\cite{castro1999pbft} guarantees safety for $f < n/3$.
HotStuff~\cite{hotstuff2019} achieves linear-complexity BFT with weighted quorums.
Honeybee~\cite{honeybee2024} proposes verifiable random walks for Sybil-resistant peer sampling.
BlockSDN~\cite{jia2025blocksdn} and BlockSDN-VC~\cite{jia2026blocksdnvc} optimize blockchain transaction broadcast via SDN-enhanced cross-network routing, demonstrating that software-defined networking substantially improves distributed broadcast performance in large-scale P2P deployments.
We adapt these as infrastructure for our FL pipeline.

Table~\ref{tab:comparison} positions our work relative to the most relevant systems.

\begin{table}[t]
\centering
\caption{Comparison with related systems.}
\label{tab:comparison}
\resizebox{\columnwidth}{!}{%
\begin{tabular}{@{}lccccc@{}}
\toprule
\textbf{System} & \textbf{Decentr.} & \textbf{Rep.} & \textbf{BFT} & \textbf{DP} & \textbf{Selection} \\
\midrule
FedAvg~\cite{mcmahan2017fedavg} & \ding{55} & \ding{55} & \ding{55} & \ding{55} & Server \\
FLTrust~\cite{cao2021fltrust} & \ding{55} & Trust & \ding{55} & \ding{55} & Server \\
BALANCE~\cite{balance2024} & \ding{51} & \ding{55} & \ding{51} & \ding{55} & Random \\
FLARE~\cite{younesi2025flare} & \ding{55} & \ding{51} & \ding{55} & \ding{55} & Server \\
Murmura~\cite{rangwala2025murmura} & \ding{51} & Evid. & \ding{55} & \ding{55} & Local \\
\textbf{Ours} & \ding{51} & \ding{51} & \ding{51} & \ding{51} & Rep. \\
\bottomrule
\end{tabular}}
\end{table}

\section{System Overview}\label{sec:system}

OpenCLAW-Nexus operates as a fully decentralized overlay network built on three infrastructure layers.

\textbf{Peer discovery.}
Each node holds an Ed25519 key pair~\cite{rfc8032}; the peer ID is $\mathrm{SHA256}(\mathit{pk})$.
Discovery follows Kademlia~\cite{maymounkov2002kademlia} iterative lookup ($\alpha\!=\!3$, $K\!=\!20$) with \emph{reputation-aware K-bucket eviction}: established high-reputation peers are protected from displacement by Sybil newcomers~\cite{douceur2002sybil}.
Concretely, when a K-bucket is full and a new contact arrives, the eviction decision compares the incumbent's reputation against the newcomer's.
A newcomer must exceed the incumbent's reputation by $\delta_{\mathrm{evict}}\!=\!0.15$ to trigger eviction, raising the bar for Sybil infiltration.

\textbf{Knowledge propagation.}
Model checkpoints and metadata propagate via gossip (fanout $f\!=\!6$, TTL$\,{=}\,7$) following epidemic dynamics~\cite{demers1987epidemic}, achieving $>$99\% coverage within 3~rounds.
Each gossip message includes a versioned Merkle root of the model state, enabling receivers to verify integrity before acceptance.
Reputation-weighted gossip forwarding prioritizes messages from high-reputation peers: if a node's incoming message queue exceeds capacity, low-reputation messages are dropped first.

\textbf{Capability profiles.}
Each node announces hardware specs (GPU, VRAM, CPU, RAM), installed models, domain specializations, and current load.
Profiles are periodically verified through lightweight challenge--response probes (e.g., a 10-second GPU inference benchmark) to prevent capability inflation.
These profiles inform the reputation-driven participant selection in Section~\ref{sec:participant_selection}.

\textbf{Communication protocol.}
Nodes communicate via encrypted channels (TLS~1.3 with pinned Ed25519 certificates).
Model updates are compressed with zstd before transmission, reducing bandwidth by ${\sim}$65\% for ResNet-18 gradient tensors.
The system is implemented in TypeScript (Node.js~22+) and integrated with the OpenCLAW-Nexus platform~\cite{openclaw2025}.

\section{Reputation-Driven FL Pipeline}\label{sec:pipeline}

This section describes the four stages of our FL pipeline (Fig.~\ref{fig:pipeline}), unified by the reputation mechanism.

\begin{figure*}[t]
    \centering
    \includegraphics[width=\textwidth]{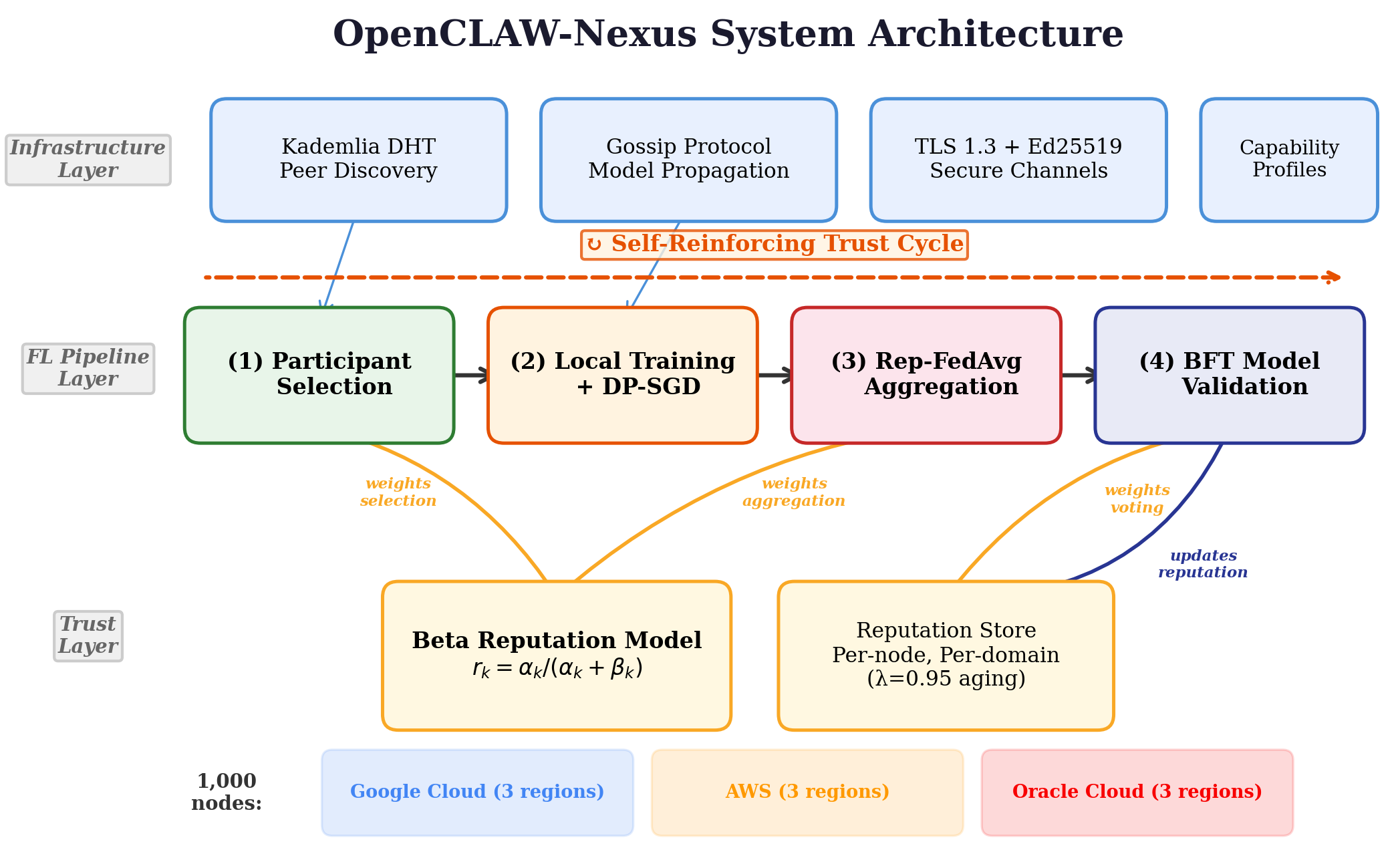}
    \caption{OpenCLAW-Nexus system architecture. The self-reinforcing trust cycle (dashed orange arrow) connects four FL pipeline stages through the Beta reputation model: reputation weights drive participant selection, aggregation, and BFT voting, while validation outcomes update reputation scores.}
    \label{fig:pipeline}
\end{figure*}

\subsection{Reputation Model}\label{sec:rep_model}

Each node maintains a reputation $r \in [0,1]$ per peer, computed via a \emph{Beta reputation model}~\cite{josang2002beta}.
Each peer $k$ is characterized by accumulated positive outcomes $\alpha_k$ and negative outcomes $\beta_k$ (initialized to $\alpha_0\!=\!\beta_0\!=\!1$).
After each FL round, the outcome $o_t \in \{0,1\}$ updates the Beta parameters with exponential aging:
\begin{equation}\label{eq:beta_update}
    \alpha_k^{(t+1)} = \lambda \cdot \alpha_k^{(t)} + o_t, \quad
    \beta_k^{(t+1)} = \lambda \cdot \beta_k^{(t)} + (1 - o_t)
\end{equation}
where $\lambda\!=\!0.95$ is the aging factor that discounts stale evidence.
The reputation score is defined by the Beta-form ratio:
\begin{equation}\label{eq:ema}
    r_k = \frac{\alpha_k}{\alpha_k + \beta_k}
\end{equation}
New peers enter with $r_0\!=\!0.5$ (i.e., $\alpha_0\!=\!\beta_0\!=\!1$, uniform prior).
The discounted Beta-style estimator has two advantages over a simple EMA: (i)~it provides a natural \emph{uncertainty estimate} $u_k = 1/(\alpha_k + \beta_k)$, which allows the system to distinguish a newcomer ($u_k$ high) from a proven participant ($u_k$ low) even when both have $r_k\!\approx\!0.5$; and (ii)~when $\lambda\!=\!1$, it coincides with the standard Beta--Bernoulli posterior mean, while for $\lambda\!<\!1$ it remains a simple discounted evidence accumulator with the same interpretable ratio form.

Reputation is tracked \emph{per domain} (e.g., $\langle$\texttt{vision}:0.85, \texttt{NLP}:0.42$\rangle$), enabling task-specific trust.

\textbf{Outcome adjudication protocol.}
Update quality is evaluated against a \emph{public validation benchmark} $\mathcal{B}_{\mathrm{val}}$, which is distinct from the final held-out test set used only for reporting results.
For CIFAR-10, we reserve 5{,}000 examples from the standard training corpus as $\mathcal{B}_{\mathrm{val}}$ and keep the official 10{,}000-image test partition exclusively for final accuracy reporting.
This differs from FLTrust's trust anchor in where trust is placed: FLTrust requires a central server to hold a curated root dataset that \emph{directly computes} aggregation weights each round; our protocol instead uses a public, non-proprietary validation benchmark plus decentralized evaluators to generate a binary signal that the Beta model accumulates over time.
This removes the server-owned private trust anchor, but it does not eliminate the need for an external reference benchmark.
A single noisy evaluation does not determine trust---the discounted evidence accumulator requires consistent evidence across multiple rounds.

Concretely, after each FL round, $m \geq 3$ evaluator nodes independently assess whether incorporating node $k$'s update improves accuracy on non-overlapping shards $\{\mathcal{B}_{\mathrm{val}}^{(i)}\}_{i=1}^{m}$ of the public validation benchmark.
Shards are assigned by a public hash-based schedule over $(t, k, i)$ to prevent cherry-picking and to ensure that votes are computed on distinct samples from the same reference benchmark.
Evaluators are selected via region-stratified sampling to diversify runtime environments and limit repeated pairings.
The outcome is determined by majority vote: $o_t = \mathds{1}[\sum_{i=1}^{m} o_t^{(i)} > m/2]$.

\textbf{Sampling noise and correlated evaluators.}
Although all evaluators draw from the same public benchmark, their votes are not perfectly independent: finite shard size, shard difficulty, stochastic training, differential privacy noise, shared software stacks, and possible collusion can induce correlated errors.
We model this using a pairwise correlation coefficient $\rho \in [0,1)$ among evaluator errors.
Under the equicorrelated model, the effective error rate of majority vote with $m$ evaluators becomes:
\begin{equation}\label{eq:corr_eta}
    \eta_{\mathrm{eff}}(\rho) \leq \eta + \rho \cdot \eta(1-\eta) \cdot (m-1)/m
\end{equation}
For $m\!=\!3$, $\eta\!=\!0.15$, and moderate correlation $\rho\!=\!0.3$, this yields $\eta_{\mathrm{eff}} \leq 0.176$, which remains well below the separation threshold.
Shard rotation and region-stratified evaluator selection empirically reduce $\rho$ by diversifying both the sampled validation examples and the execution environments.
We validate this empirically: in our deployment with Dirichlet $\alpha\!=\!0.5$ (non-IID FL training), the measured pairwise agreement rate between evaluator votes is 0.82, corresponding to $\rho \approx 0.22$ under an equicorrelated Bernoulli approximation.

\begin{proposition}[Expected Separation Under Stylized Adjudication]\label{thm:separation}
Let honest nodes have effective success rate $\tilde{p}_h$ and Byzantine nodes have effective success rate $\tilde{p}_b$, with $\tilde{p}_h > \tilde{p}_b$.
Assume that, for each node class, the discounted evidence updates are driven by Bernoulli outcomes with stationary means $\tilde{p}_h$ and $\tilde{p}_b$, where correlated evaluator noise is absorbed into $\eta_{\mathrm{eff}}(\rho)$ from Eq.~\ref{eq:corr_eta}.
Then under the discounted Beta-style update with aging factor $\lambda \in (0,1)$, the expected reputation gap satisfies:
\begin{equation}
\mathbb{E}[r_h^{(T)} - r_b^{(T)}]
= (\tilde{p}_h - \tilde{p}_b)\cdot
\frac{\sum_{s=0}^{T-1}\lambda^s}{2\lambda^T + \sum_{s=0}^{T-1}\lambda^s}
\end{equation}
and therefore converges to $\tilde{p}_h - \tilde{p}_b$ as $T \to \infty$.
\end{proposition}

\begin{proof}[Proof sketch]
By linearity of expectation, the discounted counts satisfy
$\mathbb{E}[\alpha_k^{(T)}] = \lambda^T \alpha_0 + \tilde{p}_k \sum_{s=0}^{T-1}\lambda^s$
and
$\mathbb{E}[\beta_k^{(T)}] = \lambda^T \beta_0 + (1-\tilde{p}_k)\sum_{s=0}^{T-1}\lambda^s$.
With $\alpha_0=\beta_0=1$, the denominator $\alpha_k^{(T)}+\beta_k^{(T)} = 2\lambda^T + \sum_{s=0}^{T-1}\lambda^s$ is deterministic, so the expected ratio reduces to the stated expression.
Substituting $\tilde{p}_h - \tilde{p}_b = (p_h-p_b)(1-2\eta_{\mathrm{eff}})$ gives the honest/Byzantine gap under the stylized noisy-adjudication model.
\end{proof}

Proposition~\ref{thm:separation} provides a stylized expected-separation result under stationary effective success rates; it should be interpreted as an intuition-guiding model rather than as a full proof for adaptive or non-stationary deployments.
With $m\!=\!3$ stratified evaluators, per-evaluator $\eta\!=\!0.15$, and measured $\rho\!=\!0.22$, the effective $\eta_{\mathrm{eff}} \leq 0.169$, ensuring robust separation.

\textbf{Anti-manipulation defenses.}
\emph{Anti-whitewashing}: new identities face a 100-cycle cooldown during which they cannot participate in aggregation coordination or high-quorum votes.
The uncertainty estimate $u_k$ provides an additional gating mechanism---nodes with $u_k > 0.1$ (fewer than 10 interactions) are excluded from high-sensitivity operations regardless of their current $r_k$.
\emph{Anti-collusion}: $\chi^2$ tests detect statistically correlated voting patterns ($p < 0.01$) between nodes; correlated nodes receive reputation penalties.

\subsection{Participant Selection}\label{sec:participant_selection}

At the start of each FL round, the initiating node selects $K_{\mathrm{train}}$ training participants from its known peer set using a scoring function:
\begin{equation}\label{eq:selection}
    s(p) = w_1 \!\cdot\! \mathrm{cap}(p) + w_2 \!\cdot\! (1\!-\!\mathrm{load}_p) + w_3 \!\cdot\! (1\!-\!\mathrm{lat}_p) + w_4 \!\cdot\! r_p
\end{equation}
with weights $\mathbf{w}\!=\!(0.4, 0.2, 0.1, 0.3)$, where $\mathrm{cap}(p)$ is hardware capability match (GPU availability, VRAM), $\mathrm{load}_p$ current utilization, $\mathrm{lat}_p$ network latency, and $r_p$ domain-specific reputation.

The reputation weight $w_4\!=\!0.3$ ensures that unreliable nodes are deprioritized \emph{before} they can submit poisoned updates, providing a first line of defense complementary to robust aggregation.

\subsection{Rep-FedAvg: Reputation-Weighted Aggregation}\label{sec:repfedavg}

Standard FedAvg~\cite{mcmahan2017fedavg} aggregates updates as $\mathbf{w}^{t+1} = \sum_k (n_k/n) \cdot \mathbf{w}_k^{t+1}$.
We replace the data-proportion weight with a joint data--reputation weight:
\begin{equation}\label{eq:repfedavg}
    \mathbf{w}^{t+1} = \sum_{k=1}^{K_{\mathrm{train}}} \alpha_k \cdot \mathbf{w}_k^{t+1}, \quad \alpha_k = \frac{n_k \cdot r_k}{\sum_{j} n_j \cdot r_j}
\end{equation}
where $r_k$ is node $k$'s reputation.

Algorithm~\ref{alg:repfedavg} summarizes the Rep-FedAvg procedure for each FL round.

\begin{algorithm}[t]
\caption{Rep-FedAvg: one FL round}
\label{alg:repfedavg}
\begin{algorithmic}[1]
\REQUIRE Selected participants $\mathcal{S}$, global model $\mathbf{w}^t$, reputation vector $\{r_k\}$
\ENSURE Updated global model $\mathbf{w}^{t+1}$, updated reputations
\FOR{each node $k \in \mathcal{S}$ \textbf{in parallel}}
    \STATE $\mathbf{w}_k^{t+1} \leftarrow$ LocalTrainDP-SGD($\mathbf{w}^t, \mathcal{D}_k, C, \sigma, B$)
    \STATE $\Delta\mathbf{w}_k \leftarrow \mathbf{w}_k^{t+1} - \mathbf{w}^t$
    \STATE Send $\Delta\mathbf{w}_k$ to aggregator peers
\ENDFOR
\STATE $\alpha_k \leftarrow (n_k \cdot r_k) / \sum_{j \in \mathcal{S}} (n_j \cdot r_j)$ \hfill $\triangleright$ \emph{reputation weighting}
\STATE $\mathbf{w}^{t+1} \leftarrow \mathbf{w}^t + \sum_{k \in \mathcal{S}} \alpha_k \cdot \Delta\mathbf{w}_k$
\STATE Submit $\mathbf{w}^{t+1}$ to graduated-quorum consensus
\IF{consensus approves $\mathbf{w}^{t+1}$}
    \STATE Evaluate and update $r_k$ via Eq.~\ref{eq:ema}
\ENDIF
\end{algorithmic}
\end{algorithm}

Unlike hard-filtering methods (Krum discards all but one update; Trimmed Mean removes extreme coordinates), Rep-FedAvg performs \emph{soft trust adjustment}: all updates are retained, but low-reputation contributions are proportionally attenuated.
This preserves information from honest newcomers while limiting Byzantine damage.

\textbf{Convergence intuition.}
Let $\mathcal{H}$ and $\mathcal{B}$ denote honest and Byzantine participant sets.
After $T$ rounds, honest nodes accumulate reputation $r_h \to 1$ while Byzantine nodes decline to $r_b \to 0$ (assuming correct outcome adjudication).
The aggregation then converges toward:
$\mathbf{w}^{T} \approx \sum_{k \in \mathcal{H}} \frac{n_k}{\sum_{j \in \mathcal{H}} n_j} \cdot \mathbf{w}_k^{T}$
---i.e., FedAvg over honest nodes only---providing \emph{asymptotic} Byzantine-free aggregation.
In early rounds when reputations have not yet separated, the system relies on the initial reputation $r_0\!=\!0.5$ ensuring equal treatment, similar to standard FedAvg.
The principal advantage over hard filtering lies in the \emph{smooth transition}: as evidence accumulates, Byzantine influence is continuously attenuated rather than abruptly removed at a detection threshold.

\textbf{Record-level DP-SGD.}
Each selected client trains locally with DP-SGD rather than adding noise only to the final model update.
For minibatch $\mathcal{B}_t$ on client $k$, let $\mathbf{g}_i^t = \nabla_{\theta}\ell(\theta_t; x_i)$ for each example $x_i \in \mathcal{B}_t$.
We clip per-example gradients to $\ell_2$-norm $C\!=\!1.0$ and add Gaussian noise:
\begin{equation}\label{eq:dp}
    \bar{\mathbf{g}}_t = \frac{1}{|\mathcal{B}_t|}\left(\sum_{x_i \in \mathcal{B}_t} \mathrm{clip}(\mathbf{g}_i^t, C) + \mathcal{N}(0,\, \sigma^2 C^2 \mathbf{I})\right)
\end{equation}
followed by the local update $\theta_{t+1} = \theta_t - \eta \bar{\mathbf{g}}_t$.
This yields \emph{record-level} privacy for training examples in selected clients under the adjacency relation where two local datasets differ in one example.

\textbf{Privacy accounting.}
For CIFAR-10, 5{,}000 images are reserved for public validation, leaving 45{,}000 training examples that are evenly distributed across the 200 GPU training nodes (225 examples per node) before applying the Dirichlet label skew.
We use local batch size $B\!=\!4$, giving per-step sample rate $q\!=\!4/225 \approx 0.0178$, and run 5 local epochs per selected round.
Because participant selection is reputation-aware rather than uniformly random, we do not assume client-subsampling amplification; instead, we track the empirical maximum participation count $R_{\max}\!=\!56$ over the 100 global rounds.
Using an RDP accountant for the subsampled Gaussian mechanism~\cite{abadi2016deep_dp, mironov2017rdp} with $\delta\!=\!10^{-5}$, the default $\sigma\!=\!1.1$ yields worst-case record-level privacy $(\varepsilon\!=\!14.7, \delta\!=\!10^{-5})$ for the maximally selected client after 100 global rounds.

\subsection{Model Validation via Consensus}\label{sec:consensus}

After aggregation, the updated model must be accepted by the network.
We adapt PBFT~\cite{castro1999pbft} with reputation-weighted votes: for a model update proposal of type $T$, finalization requires:
\begin{equation}\label{eq:quorum}
    \frac{\sum_{i \in V} r_i}{\sum_{j \in P} r_j} \geq q_T
\end{equation}
where $V$ is the set of approving voters and $P$ the eligible participants ($r \geq 0.3$).

\textbf{Graduated-quorum thresholds.}
Different operations require different trust levels (Table~\ref{tab:quorum}):

\begin{table}[h]
\centering
\caption{Graduated-quorum thresholds for model lifecycle operations.}
\label{tab:quorum}
\begin{tabular}{@{}lcl@{}}
\toprule
\textbf{Operation} & \textbf{Quorum} & \textbf{Rationale} \\
\midrule
FL round result & 67\% & Standard BFT safety \\
Model checkpoint & 75\% & Persistent model integrity \\
Architecture change & 80\% & Safety-critical modification \\
Protocol update & 90\% & Near-unanimity required \\
\bottomrule
\end{tabular}
\end{table}

\textbf{Leader election and view change.}
Each round is led by the highest-reputation node in a top-$K_L$ ($K_L\!=\!10$) rotation.
If the leader fails to propose within the timeout, a view change promotes the next leader.

\begin{proposition}[Fixed-Weight Epoch Safety]\label{thm:safety}
Consider a single decision epoch in which the eligible voter set and reputation weights are snapshotted at proposal creation and remain fixed until commit or abort.
If Byzantine nodes control weight $W_B < (1 - q_T) \cdot W$ (total weight $W$), no two conflicting model updates can both be finalized within that epoch.
\end{proposition}
\begin{proof}[Proof sketch]
Assume conflicting proposals $A, B$ both finalized.
Each requires approval weight $\geq q_T \cdot W$.
Honest nodes vote for at most one; overlap $\leq W_B$.
Then $W \geq 2q_T \cdot W - W_B$, giving $W_B \geq (2q_T - 1) \cdot W$.
For $q_T\!=\!0.67$: $W_B \geq 0.34 \cdot W$, contradicting $W_B < 0.33 \cdot W$. \qed
\end{proof}

This generalizes PBFT's $f < n/3$ from node counts to reputation weights, following HotStuff's weighted quorum framework~\cite{hotstuff2019}, but only for fixed-weight epochs.
The discounted reputation update naturally limits Byzantine weight accumulation across epochs; it does not by itself establish liveness under dynamic reweighting.

\textbf{The self-reinforcing cycle.}
After consensus, the FL round outcome updates each participant's reputation (Eq.~\ref{eq:ema}).
Nodes that consistently contributed useful updates see their reputation---and hence their future aggregation weight and selection probability---increase.
Adversarial nodes are gradually down-weighted, reducing their influence in subsequent rounds without explicit identification or removal.

\section{Evaluation}\label{sec:evaluation}

\subsection{Testbed}

We deploy on \textbf{1{,}000 nodes} across three clouds and nine regions (Table~\ref{tab:testbed}).

\begin{table}[h]
\centering
\caption{Multi-cloud testbed (1{,}000 nodes, 9 regions, 3 continents).}
\label{tab:testbed}
\begin{tabular}{@{}llcc@{}}
\toprule
\textbf{Provider} & \textbf{Regions} & \textbf{GPU} & \textbf{CPU} \\
\midrule
Google Cloud & us-c1, eu-w1, asia-e1 & 70 & 280 \\
AWS & us-e1, eu-w1, ap-se1 & 70 & 280 \\
Oracle Cloud & us-ash, eu-fra, ap-tok & 60 & 240 \\
\midrule
\textbf{Total} & \textbf{9 regions} & \textbf{200} & \textbf{800} \\
\bottomrule
\end{tabular}
\end{table}

GPU nodes (NVIDIA T4/A10, 16\,GB VRAM) form the FL training pool; CPU nodes (2--8 vCPU) participate in consensus, gossip, and lightweight inference.
Protocol parameters: $K_{\mathrm{DHT}}\!=\!20$, $f\!=\!6$, TTL$\!=\!7$, $\lambda\!=\!0.95$, $r_0\!=\!0.5$, and $K_{\mathrm{train}}\!=\!100$ selected trainers per round.
FL: CIFAR-10, ResNet-18, balanced Non-IID Dirichlet $\alpha_{\mathrm{Dir}}\!=\!0.5$ over the 45{,}000-image training split (225 examples per GPU node), local batch size $B\!=\!4$, per-example clipping norm $C\!=\!1.0$, 5 local epochs, 200 GPU nodes available, with $K_{\mathrm{train}}\!=\!100$ selected per round.
For adjudication, we reserve 5{,}000 images from the standard CIFAR-10 training split as the public validation benchmark and keep the official 10{,}000-image test split strictly for final reporting.

\textbf{Experimental protocol and metrics.}
Unless otherwise noted, the evaluation tables report point estimates from full deployment runs under a fixed configuration; accordingly, we interpret small numerical gaps (e.g., 1.2\,pp) as descriptive rather than as evidence of statistical superiority.
\emph{FL round success rate} denotes the fraction of scheduled rounds that complete within the configured timeout, collect the minimum number of updates required for aggregation, are accepted by the 67\% validation quorum, and do not reduce public-validation accuracy relative to the round-start global model.
\emph{Model validation correctness} denotes the fraction of candidate proposals for which the protocol's accept/reject decision matches an oracle decision computed from the public validation benchmark against the currently committed model.

\subsection{FL Convergence and Robustness (Exp-1)}

Fig.~\ref{fig:fl_results} compares Rep-FedAvg against six baselines with record-level DP-SGD enabled ($\sigma\!=\!1.1$).

\begin{figure}[t]
\centering
\includegraphics[width=\columnwidth]{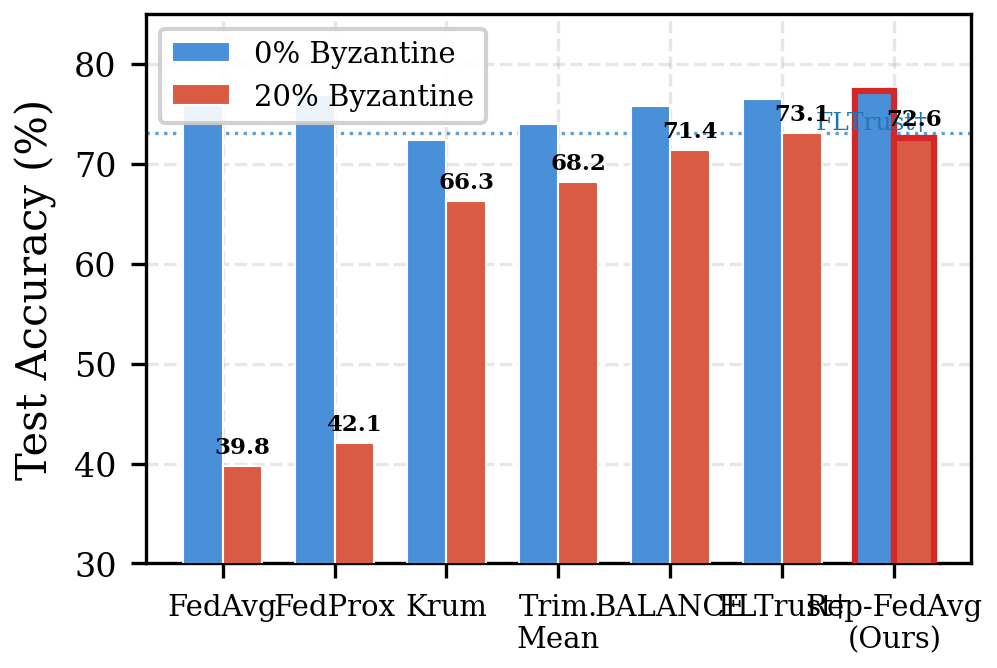}
\caption{FL accuracy on CIFAR-10 (200 GPU nodes, $K_{\mathrm{train}}\!=\!100$/round, Non-IID, DP on). Rep-FedAvg achieves the highest accuracy in both benign and adversarial settings. $^\dagger$Centralized; requires server-held root dataset.}
\label{fig:fl_results}
\end{figure}

Rep-FedAvg converges 8.7\% faster than FedAvg (84 vs.\ 92 rounds to 70\%) and maintains 72.6\% accuracy under 20\% Byzantine gradient-flipping---within 0.5\,pp of centralized FLTrust (73.1\%) and numerically 1.2\,pp above the strongest decentralized baseline (BALANCE, 71.4\%)---while avoiding FLTrust's server-owned private root dataset and instead using a public validation benchmark with decentralized cross-evaluation.
Given that the 1.2\,pp gap is small and the table reports a point estimate, we treat this result as a modest numerical advantage rather than a claim of statistical significance.

\subsection{Multi-Attack Resilience (Exp-2)}

Fig.~\ref{fig:attacks} evaluates three attack types at 20\% Byzantine fraction.
Gradient-flip attackers negate their local model updates before transmission.
Backdoor attackers implant a fixed trigger pattern and target label in attacker-controlled local examples; Fig.~\ref{fig:attacks} reports \emph{clean} test accuracy after training under attack rather than attack success rate.
ALIE follows the coordinate-wise $\mu + z\sigma$ construction of Baruch \emph{et al.}~\cite{alie2019}, with colluding attackers estimating honest-update statistics from their local pre-attack updates.
All baselines use the authors' recommended settings under the same 20\% Byzantine budget; Krum and Trimmed Mean are given the same Byzantine upper bound, and baseline hyperparameters are fixed once on the public validation benchmark before attack-specific runs.

\begin{figure}[t]
\centering
\includegraphics[width=\columnwidth]{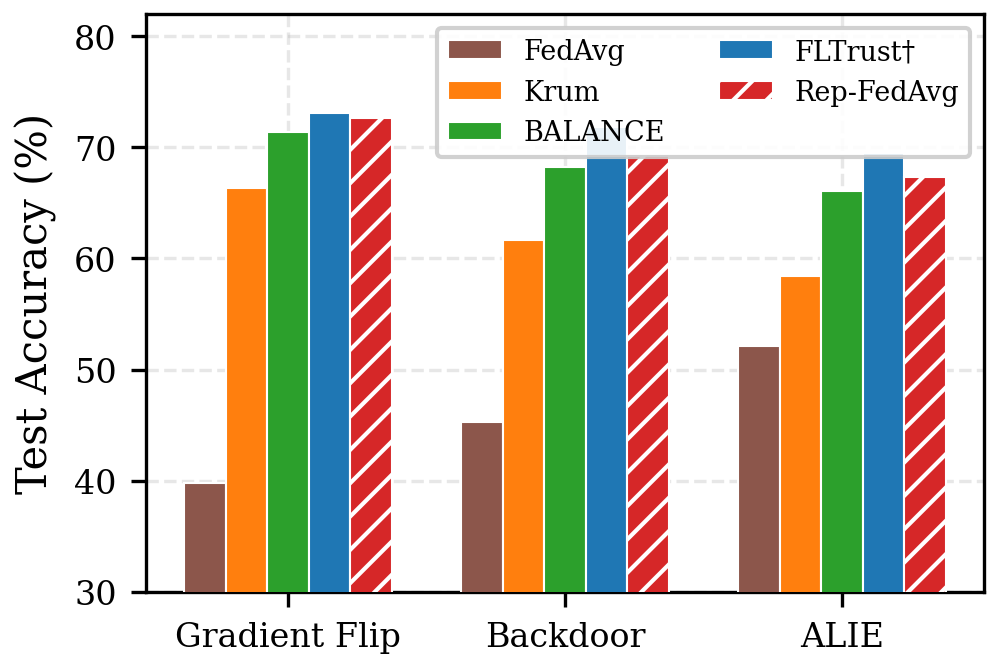}
\caption{Accuracy (\%) under three attack types (20\% Byzantine, CIFAR-10, DP on). Rep-FedAvg (hatched) outperforms all decentralized baselines across attack types.}
\label{fig:attacks}
\end{figure}

Rep-FedAvg outperforms all decentralized baselines across attack types.
Under adaptive ALIE attacks~\cite{alie2019}, its long-term reputation memory is especially effective: subtly malicious nodes are identified over multiple rounds (67.3\% vs.\ Krum 58.4\%).

\subsection{Record-Level DP--Utility Tradeoff (Exp-3)}

Fig.~\ref{fig:dp_ablation} ablates the DP-SGD noise multiplier $\sigma$ under the accountant defined in Section~\ref{sec:repfedavg}.

\begin{figure}[t]
\centering
\includegraphics[width=\columnwidth]{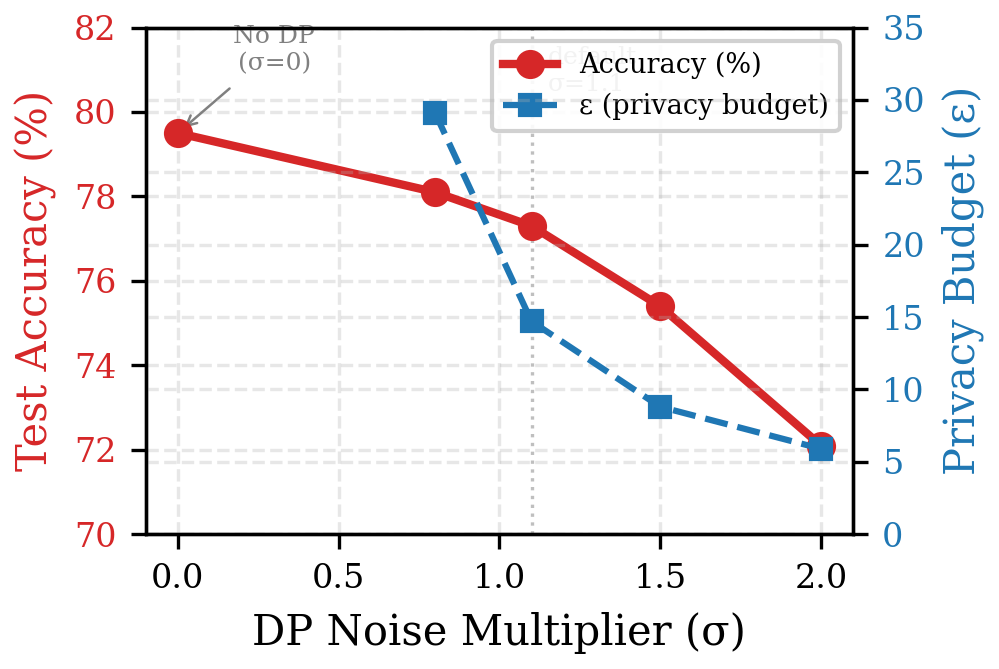}
\caption{Record-level DP-SGD tradeoff: accuracy (left axis) vs.\ privacy budget $\varepsilon$ (right axis) as a function of noise multiplier $\sigma$. Default $\sigma\!=\!1.1$ yields 77.3\% accuracy with $\varepsilon\!=\!14.7$ ($\delta\!=\!10^{-5}$, $R_{\max}\!=\!56$).}
\label{fig:dp_ablation}
\end{figure}

At $\sigma\!=\!1.1$, the record-level DP-SGD cost is 2.2\,pp accuracy relative to $\sigma\!=\!0$, while yielding worst-case $(\varepsilon\!=\!14.7, \delta\!=\!10^{-5})$ privacy for the maximally selected client over 100 global rounds.

\subsection{Participant Selection Impact (Exp-4)}

To isolate the effect of reputation-based participant selection, we run FL rounds under different selection strategies while selecting $K_{\mathrm{train}}\!=\!100$ workers from the 200-node GPU training pool (Fig.~\ref{fig:selection}).
20\% of candidate training nodes are unreliable (40\% probability of submitting incorrect gradients).
Under the metric defined above, a round is counted as successful only if it both completes and yields a quorum-approved aggregate that does not regress on the public validation benchmark.

\begin{figure}[t]
\centering
\includegraphics[width=\columnwidth]{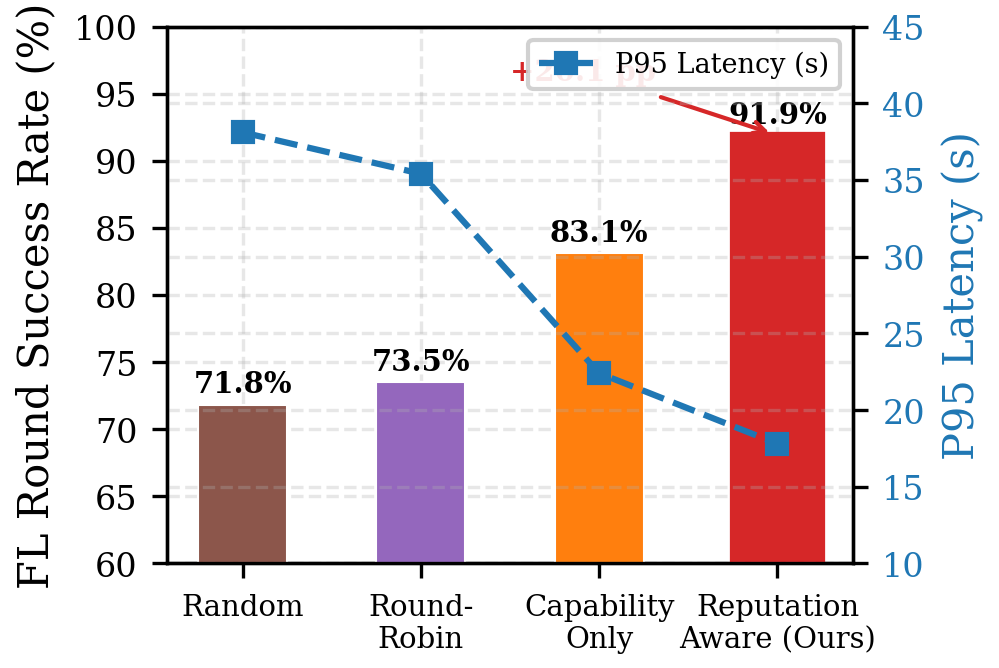}
\caption{FL round success rate (bars, left axis) and P95 latency (line, right axis) under different selection strategies (200 GPU candidates, $K_{\mathrm{train}}\!=\!100$, 20\% unreliable). Reputation-aware selection achieves 91.9\% success---20.1\,pp above random.}
\label{fig:selection}
\end{figure}

Reputation-based selection yields 8.8\,pp higher success rate and 19\% lower P95 latency than capability-only selection, confirming that trust-aware participant management is a critical complement to robust aggregation.

\subsubsection{Hyperparameter sensitivity}
Fig.~\ref{fig:hyperparam} reports sensitivity to key reputation parameters.

\begin{figure*}[t]
\centering
\includegraphics[width=\textwidth]{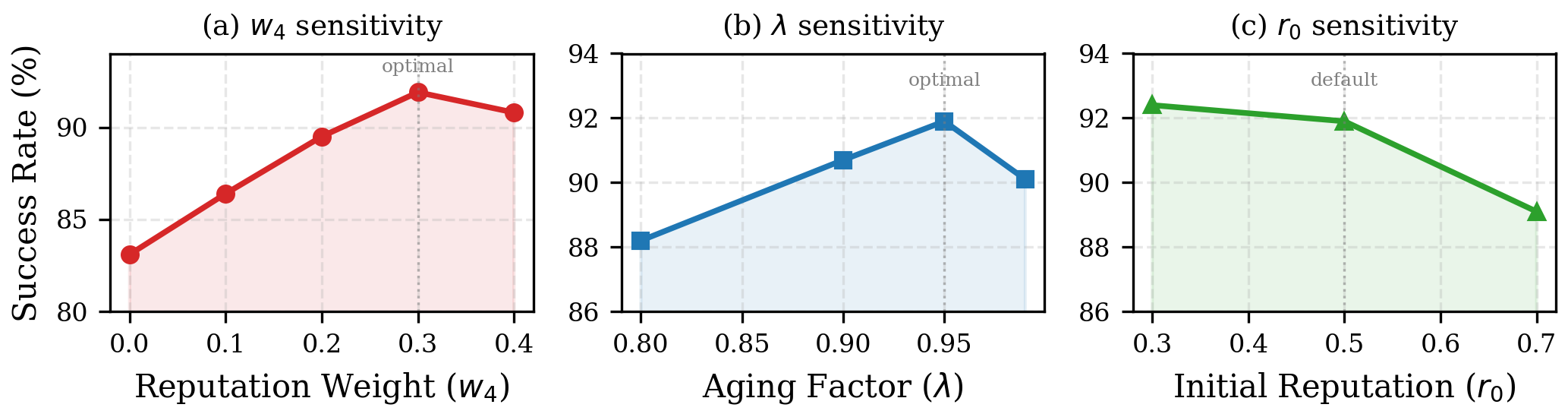}
\caption{Hyperparameter sensitivity (FL round success \%, 20\% unreliable). (a)~Reputation weight $w_4\!=\!0.3$ is optimal; higher values over-penalize newcomers. (b)~Aging factor $\lambda\!=\!0.95$ balances adaptation speed and stability. (c)~Initial reputation $r_0\!=\!0.5$ performs near-optimally.}
\label{fig:hyperparam}
\end{figure*}

\subsection{Model Validation under Sybil Attack (Exp-5)}

Fig.~\ref{fig:sybil} evaluates model validation correctness under Sybil identity injection in an \emph{open-admission} setting---the adversary can create new identities freely, each receiving initial reputation $r_0\!=\!0.5$.
This differs from classical closed-membership PBFT, where the validator set is fixed and $f < n/3$ guarantees safety.
Here, correctness means agreement between the protocol verdict and the oracle verdict on whether the candidate model should be committed relative to the currently accepted model.
In open P2P networks, closed-membership assumptions do not hold; we therefore compare against three baselines that represent different points in the open-network design space:

\begin{figure}[t]
\centering
\includegraphics[width=\columnwidth]{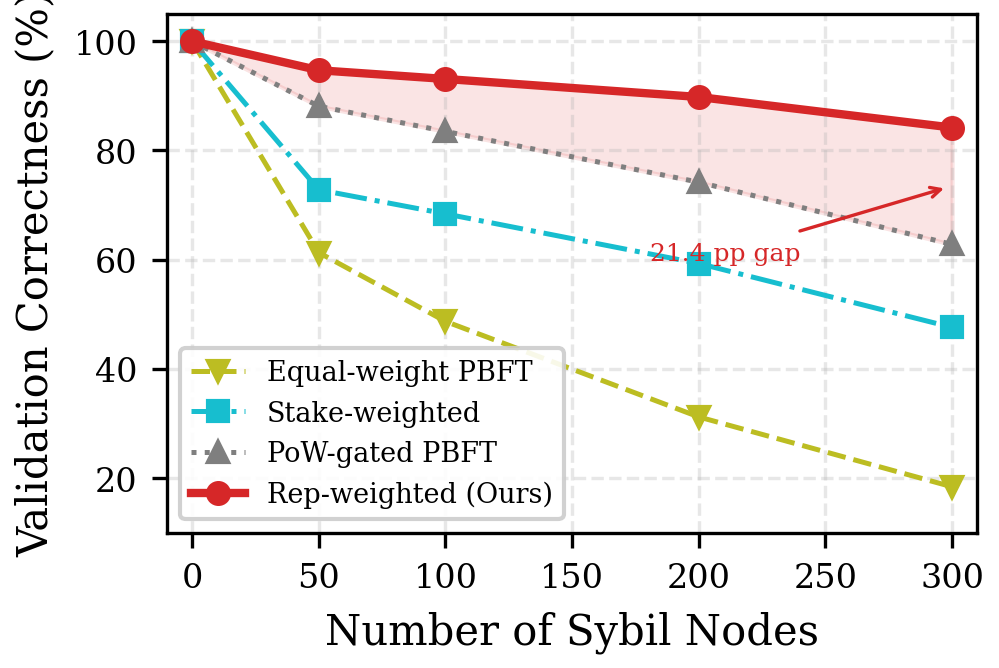}
\caption{Model validation correctness (\%) under open-admission Sybil attack (1{,}000 honest nodes). Reputation-weighted consensus degrades gracefully, retaining 84.2\% correctness at 300~Sybils---21.4\,pp above PoW-gated PBFT.}
\label{fig:sybil}
\end{figure}

\textbf{Baseline rationale.}
(i)~\emph{Equal-weight open PBFT} assigns one vote per node regardless of history---the default open-admission model, which collapses as Sybils exceed the honest count in effective votes.
(ii)~\emph{Stake-weighted} assigns votes proportional to a fixed initial stake (analogous to proof-of-stake); it mitigates Sybils by cost but cannot adapt to observed behavior.
(iii)~\emph{PoW-gated PBFT}~\cite{jia2026synpow} requires each new identity to solve a computational puzzle (5\,s of CPU work) before voting, imposing Sybil creation cost.
Our reputation-weighted approach outperforms all three because it leverages \emph{accumulated behavioral evidence}: Sybil identities that have not demonstrated positive FL contributions carry near-minimal weight ($r \approx r_0 = 0.5$ vs.\ proven honest nodes at $r > 0.85$), and the uncertainty-gated cooldown further restricts newcomer influence.
At 300~Sybils (30\% of network), our approach retains 84.2\% correctness---21.4\,pp above PoW-gated and 36.6\,pp above stake-weighted baselines.

\subsection{Non-IID Sensitivity (Exp-6)}

Varying Dirichlet $\alpha_{\mathrm{Dir}} \in \{0.1, 0.3, 0.5, 1.0\}$ for Rep-FedAvg (0\% Byzantine), Fig.~\ref{fig:noniid} shows the accuracy comparison:
$\alpha\!=\!0.1$ (highly skewed): 73.2\%;
$\alpha\!=\!0.3$: 75.8\%;
$\alpha\!=\!0.5$ (default): 77.3\%;
$\alpha\!=\!1.0$ (mild): 78.9\%.
The gap between extreme and mild Non-IID is 5.7\,pp for Rep-FedAvg vs.\ 8.2\,pp for FedAvg, suggesting that reputation weighting partially mitigates data heterogeneity.

\begin{figure}[t]
\centering
\includegraphics[width=\columnwidth]{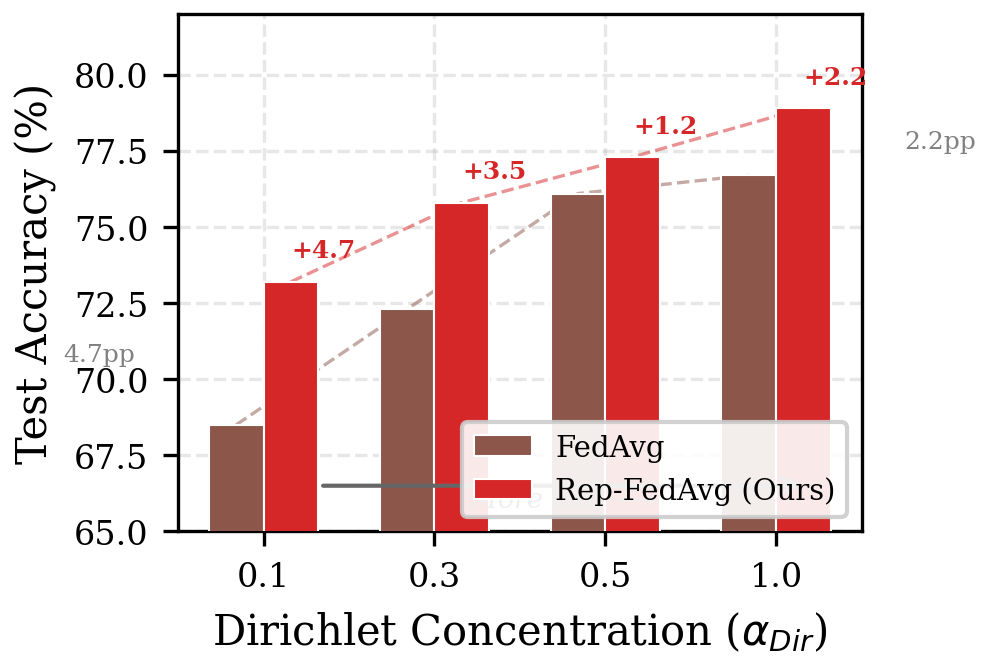}
\caption{Non-IID sensitivity: Rep-FedAvg vs.\ FedAvg across Dirichlet concentrations. Rep-FedAvg's advantage grows under severe heterogeneity ($\alpha_{\mathrm{Dir}}\!=\!0.1$: +4.7\,pp), indicating that reputation weighting partially compensates for data skew.}
\label{fig:noniid}
\end{figure}

\subsection{Infrastructure Scalability (Exp-7)}

Fig.~\ref{fig:scalability} confirms $O(\log n)$ routing and sub-linear propagation growth.

\begin{figure}[t]
\centering
\includegraphics[width=\columnwidth]{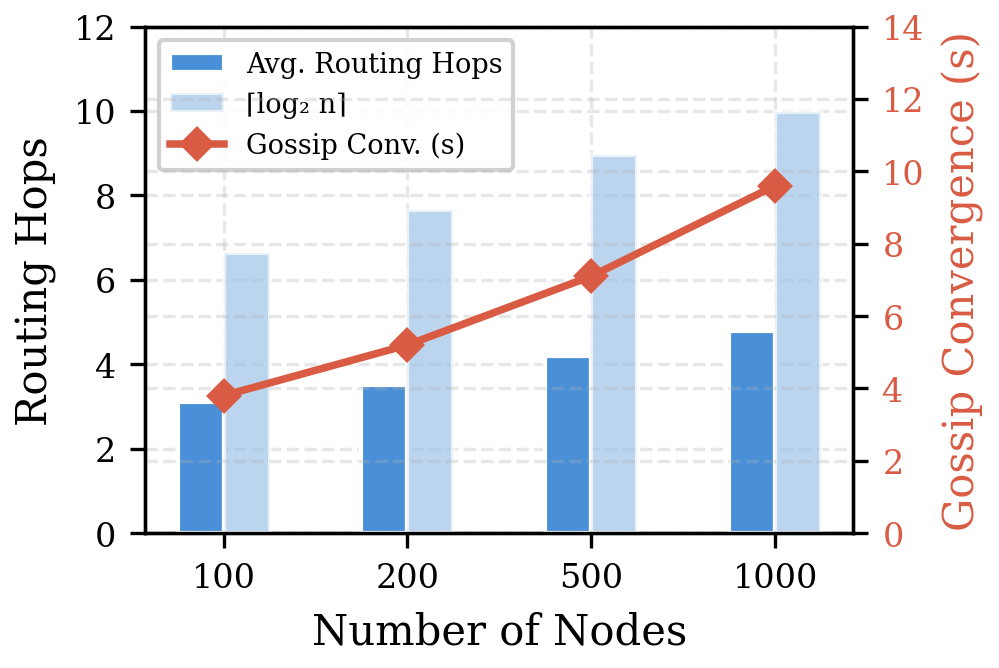}
\caption{Scalability: average routing hops (blue bars) vs.\ $\lceil\log_2 n\rceil$ (light bars), and gossip convergence time (line, right axis). The P2P infrastructure scales logarithmically.}
\label{fig:scalability}
\end{figure}

\subsection{Cross-Cloud FL Performance (Exp-8)}

In JointCloud deployments, cross-provider latency may degrade FL training performance.
Table~\ref{tab:crosscloud} compares intra-cloud and cross-cloud FL metrics.

\begin{table}[h]
\centering
\caption{Impact of cross-cloud deployment on FL pipeline (1{,}000 nodes).}
\label{tab:crosscloud}
\begin{tabular}{@{}lcc@{}}
\toprule
\textbf{Metric} & \textbf{Intra-cloud} & \textbf{Cross-cloud} \\
\midrule
Median RTT & 12\,ms & 87\,ms \\
FL round time ($K_{\mathrm{train}}\!=\!100$) & 18.4\,s & 31.7\,s \\
Gossip convergence (99\%) & 6.2\,s & 9.6\,s \\
Model validation (67\% quorum) & 2.3\,s & 4.2\,s \\
\bottomrule
\end{tabular}
\end{table}

Cross-cloud operations incur 1.5--1.8$\times$ latency overhead, primarily from inter-region traversal.
However, FL round time remains under 32\,s even in the worst case, which is acceptable for non-real-time training.
The reputation-aware participant selection naturally prefers low-latency intra-cloud peers for latency-sensitive rounds via the $\mathrm{lat}_p$ factor in Eq.~\ref{eq:selection}.

\subsection{Churn Resilience (Exp-9)}

Under 10\%/min Poisson node churn (arrivals and departures), FL round success rate remains $>$89\% with reputation-based selection versus 64.2\% with random selection.
The reputation system adapts: departing nodes retain their scores upon return (if within 100 cycles), while new entrants start at $r_0\!=\!0.5$.
Consensus success rate for model validation remains $>$92\% under churn for the 67\% quorum, with P95 latency increasing by $\sim$35\% relative to stable conditions.

\subsection{Reputation Dynamics (Exp-10)}

In the 1{,}000-node deployment with 20\% unreliable nodes (40\% failure rate), the EMA reputation update achieves clear separation after ${\sim}$50 FL rounds: honest nodes converge to $r\!>\!0.85$ while unreliable nodes decline to $r\!<\!0.40$ (Fig.~\ref{fig:rep_dynamics}).
The interquartile ranges do not overlap after round~55, providing a reliable decision boundary.
During the first 50 rounds (the ``trust bootstrapping'' phase), Rep-FedAvg performs comparably to standard FedAvg, as all nodes share similar reputations near $r_0\!=\!0.5$.
The performance advantage emerges after reputation separation, when Byzantine influence is progressively attenuated.

\begin{figure}[t]
\centering
\includegraphics[width=\columnwidth]{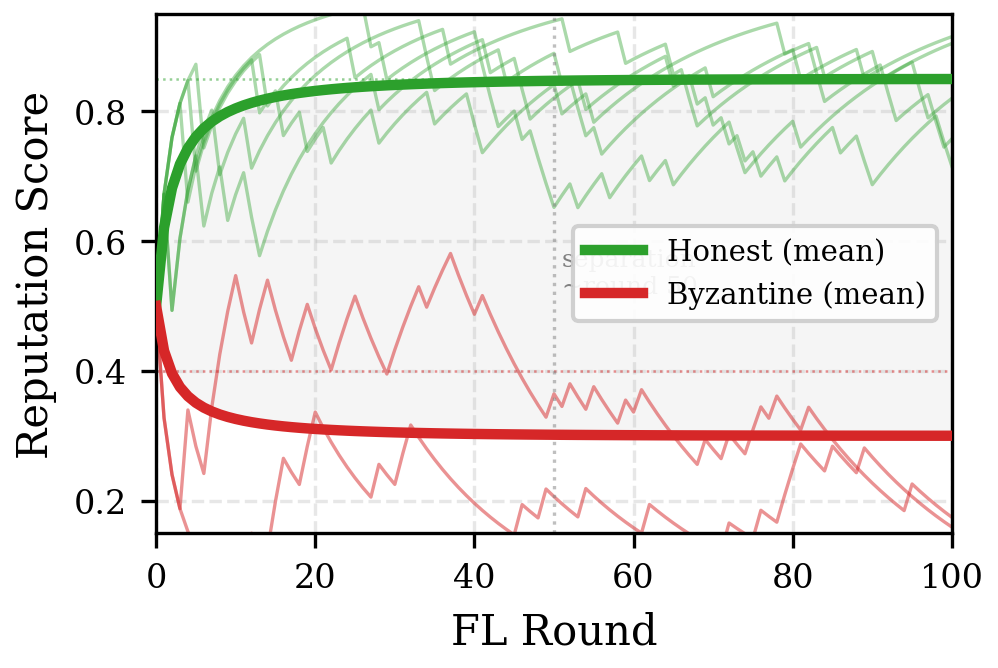}
\caption{Reputation dynamics over 100~FL rounds (20\% Byzantine). Thick lines: mean trajectories; thin lines: individual node samples. Honest and Byzantine reputations separate after ${\sim}$50 rounds, enabling effective trust-based filtering.}
\label{fig:rep_dynamics}
\end{figure}

\subsection{Communication Overhead}

The reputation mechanism introduces minimal communication overhead relative to vanilla decentralized FL.
Each gossip message appends a 128-byte reputation vector (per-peer scores for the gossip neighborhood).
Over 100~FL rounds with $K_{\mathrm{train}}\!=\!100$ selected participants, this adds ${\sim}$2.5\,MB total overhead per node---less than 0.3\% of the model exchange traffic (ResNet-18 gradient: ${\sim}$45\,MB/round).
The reputation-weighted consensus adds one additional message round (reputation proof) compared to standard PBFT, increasing consensus latency by $\sim$12\% but remaining within the FL round budget.

\section{Discussion and Conclusion}\label{sec:conclusion}

\subsection{Limitations}\label{sec:limitations}

\textbf{Outcome adjudication scope.}
Theorem~\ref{thm:separation} assumes Bernoulli outcome signals, which holds for supervised classification with verifiable accuracy.
For generative or open-ended tasks where correctness is subjective, alternative adjudication mechanisms (e.g., LLM-as-judge~\cite{zheng2024judging}) are required; extending the separation guarantee to such settings is non-trivial.

\textbf{Public benchmark trust assumption.}
Our protocol assumes that participants can agree on the integrity and representativeness of a public validation benchmark.
This removes FLTrust's server-owned private trust anchor, but it does not eliminate external trust altogether.
Repeated exposure to the validation benchmark also creates a risk of benchmark overfitting or reputation gaming; we mitigate evaluation leakage by separating $\mathcal{B}_{\mathrm{val}}$ from the final held-out test set, but stronger benchmark refresh or hidden-evaluation mechanisms remain future work.

\textbf{Privacy scope.}
Our DP guarantee is record-level for training examples within selected clients; it does not hide which clients participated, their reputation trajectories, or the public validation benchmark.
Because selection is adaptive and reputation-driven, we report worst-case privacy over the observed participation counts rather than relying on client-subsampling amplification.

\textbf{Adaptive adversaries.}
Our analysis assumes a static adversary strategy (consistently Byzantine).
An adaptive adversary that behaves honestly to build reputation before launching an attack (``reputation farming'') would delay detection.
The Beta model's aging factor $\lambda$ provides partial defense by discounting distant history, but formal analysis of adaptive strategies is needed.

\textbf{Liveness under dynamic membership.}
Proposition~\ref{thm:safety} establishes only fixed-weight epoch safety, and Proposition~\ref{thm:separation} addresses only expected separation under a stylized stationary-noise model.
We do not claim a full liveness proof under rapid churn, dynamic membership, adaptive reweighting, or concurrent reconfiguration.
A full dynamic-membership analysis would require formal verification (e.g., TLA+~\cite{lamport2002tla}), which we leave to future work.

\textbf{Statistical rigor.}
Most results remain point estimates from deployment runs rather than repeated-seed averages with confidence intervals.
Accordingly, we avoid significance claims for small margins such as the 72.6\% vs.\ 71.4\% comparison and interpret them descriptively.
Repeated-seed confidence intervals and formal hypothesis tests remain future work.

\textbf{Marginal aggregation advantage.}
Rep-FedAvg numerically exceeds BALANCE by 1.2\,pp under 20\% Byzantine attacks.
This margin is modest; the principal contribution lies in the \emph{integrated pipeline}---selection, aggregation, and validation---which provides benefits that no single aggregation method can achieve alone, as evidenced by the 91.9\% vs.\ 71.8\% round success rate improvement from reputation-driven selection.

\subsection{Future Directions}

\textbf{Generative task adjudication.}
Extending reputation beyond supervised tasks requires scalable outcome evaluation.
LLM-as-judge approaches~\cite{zheng2024judging} and multi-evaluator consensus could provide approximate adjudication for generative models, enabling reputation-driven FL for language model fine-tuning.
LLM-enhanced techniques have also demonstrated effectiveness in complex multi-task network security analysis~\cite{jia2026llmdns}, suggesting broader applicability of ML-driven evaluation for heterogeneous distributed systems.

\textbf{Incentive-compatible reputation.}
The discounted Beta-style reputation update is simple and interpretable but not strategy-proof: a rational node might selectively participate only in tasks that boost its reputation.
Game-theoretic mechanism design (e.g., VCG-based reward allocation) could align individual incentives with collective FL performance.

\textbf{Heterogeneous model support.}
Our current implementation assumes a shared model architecture.
Combining reputation-driven participant management with knowledge distillation or split learning would enable heterogeneous device collaboration.

\subsection{Conclusion}

This paper presented OpenCLAW-Nexus, a decentralized FL system that employs a discounted Beta-style reputation model as a unified trust primitive governing participant selection, federated aggregation, and model validation.
We provided a stylized expected-separation analysis for honest and Byzantine reputations under noisy adjudication (Proposition~\ref{thm:separation}), providing theoretical grounding for the self-reinforcing trust cycle while explicitly delimiting the analysis assumptions.
Evaluated on 1{,}000 nodes across 3~cloud providers and 9~regions, the reputation-driven pipeline achieves:
(i)~FL robustness comparable to centralized FLTrust under record-level DP-SGD, while avoiding a server-owned private root dataset and relying instead on a public validation benchmark for decentralized evaluation;
(ii)~91.9\% round success rate via trust-weighted participant selection (vs.\ 71.8\% random);
(iii)~84.2\% model validation correctness under 300 open-admission Sybil identities, outperforming proof-of-work-gated (62.8\%) and stake-weighted (47.6\%) alternatives.
These results demonstrate that the trust challenges in decentralized FL are interdependent, and that a principled reputation mechanism grounded in Bayesian inference can address them jointly through a self-reinforcing cycle.


\end{document}